\pdfoutput=1
\documentclass[11pt]{article}
\usepackage[utf8]{inputenc}
\usepackage{fullpage}
\usepackage{amsfonts,amssymb}
\usepackage{newpxtext}
\usepackage{hyperref}
\usepackage{authblk}
\usepackage{mathtools}
\usepackage{marvosym}
\usepackage{microtype}
\usepackage[labelfont=bf]{caption}
\usepackage{comment}
\usepackage[numbers,sort&compress]{natbib}
\usepackage{float}
\usepackage{bbold}
\usepackage{soul}
\usepackage{bm}
\usepackage{thm-restate}
\usepackage[normalem]{ulem}
\usepackage[tight]{subfigure}

\usepackage{tikz}
\usetikzlibrary{arrows,backgrounds,calc,trees}

\pgfdeclarelayer{background}
\pgfsetlayers{background,main}
\usetikzlibrary{patterns}
    
\usepackage[inline]{enumitem}
\setlist[itemize]{label=--}
\setlist[enumerate]{label=(\arabic*),labelindent=\parindent,leftmargin=*}
    
\usepackage[ruled, lined, linesnumbered, commentsnumbered, noend, ]{algorithm2e}
\usepackage[noend]{algpseudocode}

\usepackage{todonotes}
\usepackage{balance}
\usepackage{amsmath,amsthm,booktabs,color,doi,graphicx,latexsym,url,xcolor,xspace}
\usepackage{bbm}

\newcommand{\BucketALG}{\textsf{Bucketing}\xspace}
\newcommand{\BucketALGBypass}{\textsf{BucketingBypass}\xspace}
\newcommand{\OPT}{\textsf{OPT}\xspace}
\newcommand{\DET}{\textsf{DET}\xspace}

\newcommand{\OPTtau}{$\textsf{OPT}_{\tau}$\xspace}
\newcommand{\sigtau}{$\sigma_{\tau}$\xspace}

\newcommand{\convexpath}[2]{
[   
    create hullnodes/.code={
        \global\edef\namelist{#1}
        \foreach [count=\counter] \nodename in \namelist {
            \global\edef\numberofnodes{\counter}
            \node at (\nodename) [draw=none,name=hullnode\counter] {};
        }
        \node at (hullnode\numberofnodes) [name=hullnode0,draw=none] {};
        \pgfmathtruncatemacro\lastnumber{\numberofnodes+1}
        \node at (hullnode1) [name=hullnode\lastnumber,draw=none] {};
    },
    create hullnodes
]
($(hullnode1)!#2!-90:(hullnode0)$)
\foreach [
    evaluate=\currentnode as \previousnode using \currentnode-1,
    evaluate=\currentnode as \nextnode using \currentnode+1
    ] \currentnode in {1,...,\numberofnodes} {
  let
    \p1 = ($(hullnode\currentnode)!#2!-90:(hullnode\previousnode)$),
    \p2 = ($(hullnode\currentnode)!#2!90:(hullnode\nextnode)$),
    \p3 = ($(\p1) - (hullnode\currentnode)$),
    \n1 = {atan2(\y3,\x3)},
    \p4 = ($(\p2) - (hullnode\currentnode)$),
    \n2 = {atan2(\y4,\x4)},
    \n{delta} = {-Mod(\n1-\n2,360)}
  in 
    {-- (\p1) arc[start angle=\n1, delta angle=\n{delta}, radius=#2] -- (\p2)}
}
-- cycle
}

\newtheorem{theorem}{Theorem}
\newtheorem{definition}[theorem]{Definition}
\newtheorem{lemma}[theorem]{Lemma}

\title{Dependency-Aware Online Caching}

\def\BibTeX{{\rm B\kern-.05em{\sc i\kern-.025em b}\kern-.08em
    T\kern-.1667em\lower.7ex\hbox{E}\kern-.125emX}}
\begin{document}

\author[1]{Julien Dallot}
\author[2]{Amirmehdi Jafari Fesharaki}
\author[1]{Maciej Pacut}
\author[1]{Stefan Schmid}
\affil[1]{TU Berlin, Germany}
\affil[2]{Sharif University of Technology, Iran}
\date{}

  \maketitle
\thispagestyle{empty}
\setcounter{page}{0}

\begin{abstract}
        We consider a variant of the online caching problem where the items exhibit \emph{dependencies} among each other: an item can reside in the cache only if all its dependent items are also in the cache.
        The dependency relations can form any directed acyclic graph.
		These requirements arise e.g., in systems such as CacheFlow (SOSR 2016) that cache forwarding rules for packet classification in IP-based communication networks.

	First, we present an optimal randomized online caching algorithm which accounts for dependencies among the items.
	Our randomized algorithm is $O( \log k)$-competitive, where $k$ is the size of the cache, meaning that our algorithm never incurs the cost of $O(\log k)$ times higher than even an optimal algorithm that knows the future input sequence.

	Second, we consider the bypassing model, where requests can be served at a fixed price without fetching the item and its dependencies into the cache --- a variant of caching with dependencies introduced by Bienkowski et al.\ at SPAA 2017.
	For this setting, we give an~$O( \sqrt{k \cdot \log k})$-competitive algorithm, which significantly improves the best known competitiveness.
	We conduct a small case study, to find out that our algorithm incurs on average 2x lower cost.
\end{abstract}

\section{Introduction}

Performance of most computer systems today rely on how well the caching algorithms manage their caches to avoid cache misses.
Existing caching algorithms treat cached items as independent, however in some applications, items exhibit \emph{dependencies} (modeled with a directed acyclic graph) among each other: an item can reside in the cache only if all its dependent items are also in the cache.

For example, dependency-aware caching has applications in communication networks, in IP packet classification~\cite{GuptaM01} in network routers and switches (see Section~\ref{ssec:application}), where the goal is to cache the set of heavy hitter rules.
We elaborate on the setting later in this section.
In this paper, we present an algorithm with provable performance guaranties, thus proposing a rigorous, carefully analyzed alternative to the CacheFlow system~\cite{KattaARW16}.

\pagebreak
Designing caching algorithms poses interesting challenges that can be overcome with a principled approach.
In particular, as shifts in the traffic patterns (and hence also in the heavy hitter rules) may not be predictable, which calls for algorithms that operate in an online manner and adjust the cache in reaction to the traffic they see.
A well-established method to deal with uncertainty of the future is the framework of online algorithms and competitive analysis~\cite{Borodin1998}, and the classic caching algorithms were often designed and analyzed for this setting~\cite[Ch. 3, 4]{Borodin1998} ~\cite[Ch. 3]{online-book-fiat}.
Ideally, these online algorithms achieve a good competitive ratio: intuitively, without knowing the future demand, their performance is almost as good as a clairvoyant optimal offline algorithm.

To design algorithms that efficiently manage a cache in presence of dependencies,
we generalize the well-known caching problem~\cite{Sleator1985} to respect dependencies.
The dependencies can form an arbitrary directed acyclic graph, see Figure~\ref{fig:dag} for an illustration.
The objective of the algorithm is to minimize the number of fetches.

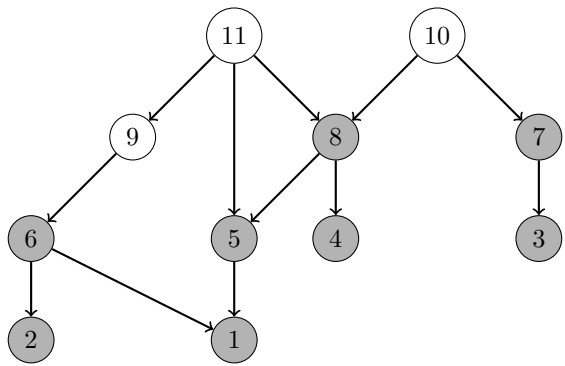
\begin{figure}[h!]
\centering

\begin{tikzpicture}[transform shape, scale=0.9]
\def\distanceBetweenNodes{1.5cm}

\tikzstyle{CachedNode}=[fill=black!30]

\node[draw, circle] at (-1*\distanceBetweenNodes,0*\distanceBetweenNodes) (node1) {$11$};
\node[draw, circle] at (1*\distanceBetweenNodes,0*\distanceBetweenNodes) (node2) {$10$};
\node[draw, circle] at (-2*\distanceBetweenNodes,-1*\distanceBetweenNodes) (node3) {$9$};
\node[draw, circle, CachedNode] at (0*\distanceBetweenNodes,-1*\distanceBetweenNodes) (node4) {$8$};
\node[draw, circle, CachedNode] at (2*\distanceBetweenNodes,-1*\distanceBetweenNodes) (node5) {$7$};
\node[draw, circle, CachedNode] at (-3*\distanceBetweenNodes,-2*\distanceBetweenNodes) (node6) {$6$};
\node[draw, circle, CachedNode] at (-1*\distanceBetweenNodes,-2*\distanceBetweenNodes) (node7) {$5$};
\node[draw, circle, CachedNode] at (0*\distanceBetweenNodes,-2*\distanceBetweenNodes) (node8) {$4$};
\node[draw, circle, CachedNode] at (2*\distanceBetweenNodes,-2*\distanceBetweenNodes) (node9) {$3$};
\node[draw, circle, CachedNode] at (-3*\distanceBetweenNodes,-3*\distanceBetweenNodes) (node10) {$2$};
\node[draw, circle, CachedNode] at (-1*\distanceBetweenNodes,-3*\distanceBetweenNodes) (node11) {$1$};

\draw[->, thick] (node1) -- (node3) {};
\draw[->, thick] (node1) -- (node4) {};
\draw[->, thick] (node1) -- (node7) {};
\draw[->, thick] (node2) -- (node4) {};
\draw[->, thick] (node2) -- (node5) {};
\draw[->, thick] (node3) -- (node6) {};
\draw[->, thick] (node4) -- (node8) {};
\draw[->, thick] (node4) -- (node7) {};
\draw[->, thick] (node5) -- (node9) {};
\draw[->, thick] (node6) -- (node10) {};
\draw[->, thick] (node6) -- (node11) {};
\draw[->, thick] (node7) -- (node11) {};

\end{tikzpicture}

\caption{Example directed acyclic graph of dependencies among items.
  An arrow from $u$ to $v$ means that if $u$ is in the cache then also $v$ must be in the cache.
  The size of the cache is $k = 8$, the items $9$, $10$ and $11$ are not in the cache while the grayed items $1,2, \ldots, 8$ are in the cache.
  If a request to $9$ would arrive, we must choose what item to evict while keeping a feasible cache ($8$ or $7$ are the only choices here) before we fetch $9$.
}

\label{fig:dag}
\end{figure}

Introducing dependencies among the cached items unravels algorithm design challenges unseen in the classic caching problem.
Consider the situation when a cache miss occurs, and we need to choose a single item to evict.
Then, we cannot freely pick any item from the cache to evict, as other items in the cache may depend on them.
Hence, simply using a classic dependency-unaware algorithm such as Random Mark~\cite{FiatKLMSY91} leads to infeasible solutions, as other items may depend on the item chosen to be evicted.
Furthermore, how to decide which item to evict? If we have two items that could be evicted, should we base the decision upon the number of descendants of these items? If we favor to evict an item with more descendants, then we would have more flexibility to choose items for future evictions.
These questions become even more challenging when items can have multiple parents, and the dependency graph is not a tree.

The main contribution of our paper is an optimal randomized online \emph{Bucketing} for the dependency-aware caching.
Our analysis shows that our generalized algorithm is $O(\log k)$-competitive, where $k$ is the size of the cache.
We dive deeper into the analysis of our randomized algorithm with a parameterized analysis of the competitive ratio, revealing how the topology of the dependency graph influences the resulting competitive ratio, linking its competitiveness to the maximum independent set in the dependency graph.

As an additional contribution, we generalize our randomized algorithm to the setting with bypassing~\cite{EpsteinILN15,Irani02a}, where requests can be served for a fixed price without fetching to the cache.
In the context of packet classification, bypassing delegates classification of a packet to the centralized controller, similarly to the design introduced in CacheFlow~\cite{KattaARW16}.
This variant of caching with dependencies and bypassing was introduced by Bienkowski et al.\ at SPAA 2017, where the authors gave an $O(k \cdot h(T))$-competitive \cite{BienkowskiMPSS17} where $h(T)$ is the height of the dependency graph $T$.
By generalizing our Bucketing procedure, we design an algorithm that is $O( \sqrt{k \cdot \log k})$-competitive.
Our result hence significantly improves the best known competitiveness, reaching the sublinear dependency on the cache size~$k$.

\subsection{Motivation: Packet Classification}
\label{ssec:application}

Packet classification is a~fundamental task in communication networks~\cite{GuptaM01}.
For example, each packet incoming to a router is matched against a set of predefined rules (e.g., does the packet match an IP range?) to determine how the router should handle the packet (e.g., send via a port);
a packet may match multiple rules, and among them, the router uses the highest priority rule.
The router then handles the packet according to the action associated with the rule.

The number of rules a router needs to store is growing rapidly for several reasons~\cite{CittadiniMUBFM10}.
This introduces significant memory requirements, requiring more and more expensive and power-hungry hardware~\cite{SarrarUFSH12}. 
To address this problem, a~natural approach, studied in systems such as CacheFlow~\cite{KattaARW16} and TreeCaching by Bienkowski et al.\ at SPAA 2017~\cite{BienkowskiMPSS17}, is to store only a small subset of the rules at the router and the rest 
in a cheaper but potentially slower memory, e.g., at a (software-defined) network controller, essentially a two-level caching hierarchy.
However, we cannot blindly use existing caching algorithms due to dependencies that arise among the rules with overlapping patterns.
Precisely, to assure that the router correctly forwards packets with its cached subset of rules, if two rules overlap, the higher priority rule must be present in the cache when the lower priority rule is present.

The structure of the dependency graph depends on how general the packet forwarding rules are. In case of prefix rule matching for a single field (destination IP), the dependency graph is a~tree~\cite{BienkowskiMPSS17}. However, in multi-field matching, commonly used in OpenFlow~\cite{McKeownABPPRST08}, the dependency can have a more general form of a DAG, even if the IP fields are matched by the longest prefix rule. Also, even for single-field matching, the wildcard rule can result in DAG dependencies~\cite{KattaARW16}.

For more background and details on such network architectures in general and on the technical setup of caching classification rules, we refer to prior works~\cite{KattaARW16, RottenstreichKJ22, BienkowskiMPSS17, AddankiPPRSV22}.

\subsection{Related Work}

Due to their wide applications, algorithms for caching (often also called paging) were studied for decades, and here we overview the results from the perspective of competitive analysis.
The seminal paper of Sleator and Tarjan~\cite{Sleator1985} originated the concept of competitive analysis of online algorithms.
In their paper, an upper bound of $k$-competitiveness was established for a family of marking algorithms that included commonly studied Least Recently Used and FIFO algorithms, and it was shown that no deterministic algorithm can be better than $k$-competitive.
Randomization helps: Fiat et al.~\cite{FiatKLMSY91} showed that an algorithm Random Mark is $2H_k$-competitive and no randomized algorithm can be better than $H_k$-competitive. Two algorithms~\cite{McGeochS91,AchlioptasCN00} match this lower bound, hence competitiveness of online caching is fully understood in this model.

Our paper is the most related to work of Bienkowski et al.~\cite{BienkowskiMPSS17}, where they introduce \emph{online tree caching}, a caching variant that can be viewed as online dependency-aware caching with the dependency graphs restricted to binary trees. In their model the requests can be \emph{bypassed}~\cite{EpsteinILN15,Irani02a}:
the request for item not present in the cache can be served from the slow memory, incurring a fixed cost.
The algorithm presented in their paper attains the competitive ratio of $O(k\cdot h(T))$, where~$h(T)$ is the height of the dependency tree.

In the context of packet classification in communication networks, an algorithm CacheFlow~\cite{KattaARW16} was proposed to deal with dependencies among the rules.
The algorithm splits the packet classification rules to minimize overlap, and uses the estimated past rule popularity statistics to determine the best cache configuration.
Empirical evaluations of CacheFlow demonstrate the applicability of the approach in the context of packet classification.
Other worth-mentioning attempts to improve caching forwarding rules tried to avoid involving the controller by using cooperation between switches~\cite{RottenstreichKJ22} or optimizing rules storage with dynamic compression algorithms \cite{FIBWithoutChurn, FIBSkiRental} for more restricted scenarios.

Some existing work on caching uses similar terminology to dependencies, but the models differ substantially from ours.
First, dependency-aware caching was studied from the practical perspective in the context of parallel processing systems~\cite{YuZWZL22}, but the dependencies are required only at the fetch time, and can immediately be evicted afterward.
Second, online caching was considered under a restriction called access graph~\cite{BorodinIRS95}, a request at time $t$ can be followed only by a subset of requests at $t+1$, consistently with a given access graph. In their work, the cache state is unrestricted, and in contrast, in our model we restrict feasible cache configurations rather than feasible input sequences.

\subsection{Contributions}

  First, we consider the dependency-aware caching problem in the most natural setting \emph{without bypassing}, where after requesting an item, it must be placed in the cache along with its dependencies.
  For this setting, we develop optimal deterministic and randomized algorithms.
  The optimal deterministic algorithm, Recursive LRU is $k$-competitive, and the randomized algorithm Bucketing is $2 H_k$-competitive, where $H_k$ is the $k$-th Harmonic number.
  Further, we characterize how the competitive ratio depends on the topology of the dependency graph.
Our randomized algorithm is $2 H_{\min\{k, \ell\}}$-competitive, where $\ell$ is the size of the maximum independent set in the transitive closure of the dependency graph, $k$ is the size of the cache, and $H_k$ is the $k$-th Harmonic number.

We complement this result by showing that no randomized algorithm can be better than $H_{\min\{k, \ell\}}$-competitive.
Hence, our algorithm is asymptotically optimal from the competitive standpoint.
We highlight that the lower bound holds even for the simplest dependency structure of the tree, hence the algorithm is optimal for the simplest case of prefix rule matching.
 
  Next, we consider a related variant of dependency-aware caching where requests can be \emph{bypassed}~\cite{EpsteinILN15,Irani02a}, meaning that for a fixed cost an algorithm could avoid potentially costly fetch of the requested item and its dependencies. 
This setting was already studied in the literature in the context of caching packet classification rules~\cite{GuptaM01}: Bienkowski et al.~\cite{BienkowskiMPSS17} proposed an $O(k\cdot h(T))$-competitive algorithm, where $h(T)$ is the height of the dependency tree (their algorithm restrict dependencies to trees). We significantly improve upon this result by developing a $(6 \sqrt{k \cdot H_{\min\{k,\ell\}}})$-competitive algorithm, a ratio that is independent of the height of the dependency graph.
We note that in contrast to this result, the competitive ratio of our algorithm can only improve when we account for properties of the dependency graph.

Finally, we perform an empirical case study comparing the TreeCaching algorithm of Bienkowski et al.~\cite{BienkowskiMPSS17} with our randomized algorithm, finding that it on average performs 2x better in terms of the average cost per request.

\section{Organization}

We explain in greater details the problem we solve and introduce notations used throughout the whole paper in section \ref{sec:preliminaries}.
We present our main randomized algorithm and prove the guaranties on its competitiveness in section \ref{sec:main_section_without_bypassing} --- we also present there a deterministic algorithm and prove the optimality of both algorithms with lower bounds.
We present a randomized algorithm in the framework where bypassing is authorized in section \ref{sec:bypassing}.
Finally, we provide empirical evidence that our randomized algorithm with bypassing authorized outperforms the best known online algorithm against Zipfian and Exponential distributions.

\section{Preliminaries}
\label{sec:preliminaries}

\subsection{Model}
\label{sec:model}
We introduce the \emph{online dependency-aware caching} problem, defined as follows.
Our task is to manage a two level memory hierarchy, consisting of a slow memory which stores the universe $\mathcal{U}$ of $n$ items, and a fast memory which stores at most $k$ items, where $k \ge 1$ is a parameter. 

At the beginning we are given a set of dependencies among the items, given as an arbitrary directed acyclic graph (DAG) $G = (\mathcal{U}, E)$, which restricts the set of feasible caches.
For each item $x$, the set of its \emph{dependencies} are the items reachable from $x$ in $G$, excluding the item $x$ itself.
At any time, an item can be present in the cache only if all its dependencies are in the cache.
We assume that each item has at most $k-1$ dependencies.
If $G$ has no edges, the problem is equivalent to the classic online caching problem~\cite{SleatorT85}.

In the online manner, we receive a~sequence $\sigma$ of requests to the items.
If a requested item is not in the cache, we must \emph{fetch} this item into the cache alongside with any of its dependencies that may not be in the cache.
As the size of the cache is limited, we may need to evict other items to fetch the requested item and its dependencies.
The goal of the algorithm is to minimize the number of fetches.

\subsection{Notation}
\label{ssec:notation}

Let $x$ and $y$ be two items of the universe $\mathcal{U}$.
We say that $y$ is a \emph{descendant} of $x$ (or $x$ is an \emph{ancestor} of $y$) if there exists a directed path from $x$ to $y$ in $G$.
By $T(x)$ we denote the set of all descendants of $x$.
In other words, $T(x)$ is the set that contains $x$ and all its dependencies.

For the rest of the paper, we fix a total order $\tau$ among the items such that $\tau$ is consistent with an arbitrary reversed topological order of the items in $G$.
For example, the numbers in the items of Figure \ref{fig:dag} are a reversed topological order and the order $\le$ on the item's associated numbers is a~valid order $\tau$ between the items.

For any directed acyclic graph $D$, we define its maximum independent set as follows. First, we take the transitive closure $D^T$ of $D$.
Then, we take the symmetric closure $D^S$ of~$D^T$.
We say that the maximum independent set of $D$ is the maximum independent set of~$D^S$.

\section{Dependency-Aware Caching without Bypassing}
\label{sec:main_section_without_bypassing}

In this section, we study the problem of online caching with dependencies in the randomized setting.
We present a randomized algorithm \BucketALG and prove that its expected competitive ratio is always below $2H_{\min \{k, \ell\}}$.
In Appendix~\ref{apx:lb} we show an asymptotically tight lower bound which proves \BucketALG is the best possible.

\subsection{The Randomized Algorithm \BucketALG}

In this section we present our randomized algorithm called \BucketALG.
The design of the algorithm relies on the notion of \emph{bucket} define hereafter.

\begin{definition}
A bucket is a subset of items with the following properties:
\begin{itemize}
  \item it is not empty
  \item all its items are cached
  \item evicting its item with maximum $\tau$ (called maximum item) leaves a feasible cache
\end{itemize}
\end{definition}

\noindent We say that a bucket is \emph{frozen} if it used to be a bucket but does not match the definition anymore.\\

Existence of a bucket ensures that one of its items --- the one with maximum $\tau$ --- is ready for eviction.
Hence, our algorithm \BucketALG uses buckets as its only eviction interface; its main tasks then are to maintain a pool of buckets and choose a bucket when an eviction is needed. 
Here follows a~description of the algorithm.

\BucketALG operates in \emph{phases}.
The first phase starts with the first request $\sigma_{1}$.
At the start of a~phase, \BucketALG generates a new pool of buckets.
To do so, it iterates over each maximal cached item $x$ (i.e., cached item without cached ancestor) and inserts $T(x)$ as a bucket in the pool (see an illustration in Figure \ref{fig:buckets}).
\BucketALG now has a non-empty pool of buckets.

\begin{figure}[h!]
\centering

\begin{tikzpicture}[transform shape, scale=1]
\def\distanceBetweenNodes{1.5cm}

\tikzstyle{CachedNode}=[fill=black!30]

\node[draw, circle] at (-1*\distanceBetweenNodes,0*\distanceBetweenNodes) (node1) {$11$};
\node[draw, circle] at (1*\distanceBetweenNodes,0*\distanceBetweenNodes) (node2) {$10$};
\node[draw, circle] at (-2*\distanceBetweenNodes,-1*\distanceBetweenNodes) (node3) {$9$};
\node[draw, circle, CachedNode] at (0*\distanceBetweenNodes,-1*\distanceBetweenNodes) (node4) {$8$};
\node[draw, circle, CachedNode] at (2*\distanceBetweenNodes,-1*\distanceBetweenNodes) (node5) {$7$};
\node[draw, circle, CachedNode] at (-3*\distanceBetweenNodes,-2*\distanceBetweenNodes) (node6) {$6$};
\node[draw, circle, CachedNode] at (-1*\distanceBetweenNodes,-2*\distanceBetweenNodes) (node7) {$5$};
\node[draw, circle, CachedNode] at (0*\distanceBetweenNodes,-2*\distanceBetweenNodes) (node8) {$4$};
\node[draw, circle, CachedNode] at (2*\distanceBetweenNodes,-2*\distanceBetweenNodes) (node9) {$3$};
\node[draw, circle, CachedNode] at (-3*\distanceBetweenNodes,-3*\distanceBetweenNodes) (node10) {$2$};
\node[draw, circle, CachedNode] at (-1*\distanceBetweenNodes,-3*\distanceBetweenNodes) (node11) {$1$};

\draw[->, thick] (node1) -- (node3) {};
\draw[->, thick] (node1) -- (node4) {};
\draw[->, thick] (node1) -- (node7) {};
\draw[->, thick] (node2) -- (node4) {};
\draw[->, thick] (node2) -- (node5) {};
\draw[->, thick] (node3) -- (node6) {};
\draw[->, thick] (node4) -- (node8) {};
\draw[->, thick] (node4) -- (node7) {};
\draw[->, thick] (node5) -- (node9) {};
\draw[->, thick] (node6) -- (node10) {};
\draw[->, thick] (node6) -- (node11) {};
\draw[->, thick] (node7) -- (node11) {};

\begin{pgfonlayer}{background}
	\fill[fill=orange!80, pattern=north west lines, opacity=0.3] \convexpath{node4, node8, node11, node7}{13pt};
	\fill[fill=orange!80, pattern=north west lines, opacity=0.3] \convexpath{node6, node11, node10, node6}{13pt};
	\fill[fill=orange!80, pattern=north west lines,  opacity=0.3] \convexpath{node5, node9}{13pt};
\end{pgfonlayer}

\end{tikzpicture}
\caption{Example of buckets formed at the beginning of a phase.
We have three buckets, depicted by dashed sets, with maximum items (i.e., candidates for eviction) $6$, $8$ and $7$. The item $1$ is in two buckets.
}

\label{fig:buckets}
\end{figure}
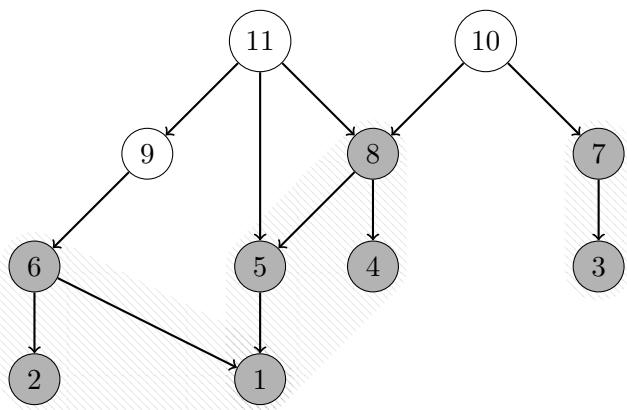

Whenever a request arrives, \BucketALG starts by removing the requested item and its dependencies from all buckets.
To satisfy the request, multiple items might be missing: in order to fetch them, \BucketALG considers the missing items one by one in increasing $\tau$.
When it considers a missing item to fetch, the cache may however be full.
In that case, \BucketALG chooses one bucket uniformly at random in the pool and evicts its maximum item out of the cache (and also remove the evicted item from all buckets).
\BucketALG then fetches the considered item.

At some point, it may happen that some subset of items in our pool no longer constitute a~bucket.
The reason may be that the subset has become empty, or that its maximum item cannot be evicted due to dependency relations.
In that situation, we say that the bucket is \emph{frozen}.
Whatever the reason that caused it, if a bucket froze then \BucketALG immediately removes that bucket from the pool.
Finally, a~new phase begins when the pool becomes empty.

Above we gave an almost complete presentation of \BucketALG.
There are however still edge cases that need to be properly addressed for correctness, for instance what to do if a new phase starts in the middle of serving a request.
For a complete description we refer to the pseudo-code in algorithm \ref{alg:random_mark_no_bypassing}.

\SetKwFor{When}{when}{do}{endwhen}
\SetKwProg{Proc}{procedure}{:}{}
\SetKwProg{Alg}{algorithm}{:}{}

\begin{algorithm}[h]
 \caption{Randomized algorithm \BucketALG}
 \label{alg:random_mark_no_bypassing}

\Proc{{ResetBuckets()}}{

	$\Gamma := \emptyset$

	\ForEach{cached item $x$ with no cached ancestor}
	{
		$\Gamma := \Gamma \cup \{ T(x) \}$
	}
	return $\Gamma$\\
}
 \Comment{Fetch $w$, evict an item beforehand if necessary}

 \Proc{{EvictAndFetch($w$, $\Gamma$)}}
 {
   \uIf {$w$ not in the cache}
   {
	 \uIf {the cache is full}
	 {
	   choose $B \in \Gamma$ uniformly at random, let $y \in B$ with maximum $\tau$\\
	   Remove $y$ from all buckets\\
	   Remove all frozen buckets from $\Gamma$\\
	   Evict $y$ from the cache\\
	 }
	 Fetch $w$ into the cache\\
   }
   return $\Gamma$\\
 }

 \Comment{The Main Algorithm}

 \Alg{{\BucketALG}}{
   $\Gamma \gets \text{ResetBuckets}()$\\

\When{a new request $v$ arrives}
{
	\For{$w$ in $T(v)$ ordered in increasing $\tau$}
	{
	  Remove the items of $T(v)$ from all buckets\\
      Remove all frozen buckets from $\Gamma$\\
                \uIf{$\Gamma = \emptyset$}
                {
                        $\Gamma \gets \text{ResetBuckets}()$\\
                        Remove the items of $T(v)$ from all buckets\\
                        Remove all frozen buckets from $\Gamma$\\
                }
		$\Gamma \gets \text{EvictAndFetch}(w, \Gamma)$\\
	}
	Serve the request to $v$\\
}
}
\end{algorithm}

If $G$ has no edges, our problem is equivalent to classic caching.
In that case, it is worth noting that our algorithm \BucketALG is then equivalent to the well-known, asymptotically optimal Random Mark algorithm by Fiat et al.~\cite{FiatKLMSY91}.
In a sense, a bucket in our algorithm directly generalizes an unmarked nodes in the random-mark algorithm.

\subsection{Proof for the competitiveness of \BucketALG}

In this section, we prove that \BucketALG holds an expected competitive ratio of $H_{\min\{k, \ell\}}$ against the oblivious adversary.\\

\subsubsection{Upper Bound for the \BucketALG Algorithm}
\label{sssec:ub_bucketalg}

In this subsection, we derive an upper bound on the expected cost (number of fetches) of \BucketALG during any phase.
Fix any input sequence $\sigma$ and directed acyclic graph $G$.
In the next paragraph we make our analysis framework more precise. 

In order to synchronize our analysis with the actions of \BucketALG, we will use a more fine-grained yet equivalent request sequence than $\sigma$.
We call this sequence $\sigma_{\tau}$ and construct it as follows.
Iterate over the requests of $\sigma$ and replace each $x$ by $T(x)$ sorted in increasing $\tau$.
We refer to the items of $\sigma_{\tau}$ as \emph{pseudo-requests}.
As \BucketALG considers the items to fetch in the same order as in~$\sigma_{\tau}$, the following correspondence holds: to any instant when \BucketALG considers a given item $w$ corresponds a unique pseudo-request $w \in \sigma_{\tau}$.
We can therefore without ambiguity refer to specific actions that \BucketALG performed after receiving a given pseudo-request $w \in \sigma_{\tau}$, such as an eviction or a fetch.

We partition the sequence $\sigma_{\tau}$ into subsequences separated by bucket regeneration events (as defined above).
We call those subsequences \emph{phases} and will upper-bound the cost of \BucketALG on each of them.

Let $P$ be a phase, and let $m$ be the number of buckets that \BucketALG generated at the beginning of $P$.
We further break $P$ into subsequences called \textit{fragments}.
The fragments make up a partition of~$P$ so that the number of buckets in the pool stays the same during each fragment.
In other words, the fragments are separated by freeze events.
If two or more buckets freeze at the same time (i.e., due to the same pseudo-request), we order these freeze events arbitrarily.
Hence, we can without ambiguity name each fragment after the number of buckets that are already frozen as it begins: we call the fragments $F_{0}, F_{1} \dots F_{m-1}$.

Let $i \in [0, m-1]$ and $b \in [1, m]$, we define $X_{i}^{b}$ as the random variable equal to the number of items that were evicted during $F_{i}$ in the $b$-th bucket to freeze.
$X_{i} = \sum_{b=1}^{m} X_{i}^{b}$ is the total number of evicted items during $F_{i}$.
The following lemma then holds.

\begin{lemma}
	\label{lem:basic_upper_bound_without_bypassing}
	$\forall i \in [0, m-1]$, $\forall b \in [i+1, m]$ it holds that
        \begin{align*}
		\mathbb{E}[X_{i}^{b}] = \frac{\mathbb{E}[X_{i}]}{m-i}
        \end{align*}
\end{lemma}

\begin{proof}
	We first show that, during any fragment $i$, each non-frozen bucket experiences the same number of evictions in expectation.
	To this end, we label the successive pseudo-requests in fragment~$F_{i}$ with integer numbers $t = 1,2,3,\ldots$.
	We define the following random variables: for any~$t$, let $\mathcal{E}(t)$ be $1$ if the fragment $F_{i}$ has a $t$-th pseudo-request and if that $t$-th pseudo-request leads to an eviction, and $0$ otherwise.
	We also define random variables $X_{i}^{b}(t)$ which equals $1$ if an eviction occurs in the $b$-th bucket when pseudo-request $t$ arrives, and $0$ otherwise.
	(Notice that~$X_{i}^{b}(t)$ equals~$0$ if the fragment $F_{i}$ does not have a $t$-th pseudo-request.)

	Then it holds that
	\begin{align*}
		\mathbb{E}[X_{i}^{b}(t)] & =  \mathbb{E}[X_{i}^{b}(t) | \mathcal{E}(t) = 1] \cdot \mathbb{P}[\mathcal{E}(t) = 1] + \mathbb{E}[X_{i}^{b}(t) | \mathcal{E}(t) = 0] \cdot \mathbb{P}[\mathcal{E}(t) = 0]\\
					 & =  \frac{1}{m-i} \cdot \mathbb{P}[\mathcal{E}(t) = 1] + 0
	\end{align*}
	Hence, the expected value of $X_{i}^{b}(t)$ does not depend on the index~$b$.
	Therefore, for all $b, b^{\prime} \in [i+1, m]$, it holds that
		$\mathbb{E}[X_{i}^{b}(t)] = \mathbb{E}[X_{i}^{b^{\prime}}(t)]$.

	Notice that for all $b \in [i+1, m]$ it holds $X_{i}^{b} = \sum_{t} X_{i}^{b}(t)$.
	Hence, each non-frozen bucket indeed in expectation experiences the same number of evictions.
	More formally, for all $b, b^{\prime} \in [i+1, m]$, we have $
		\mathbb{E}[X_{i}^{b}] = \mathbb{E}[X_{i}^{b^{\prime}}]$.

	Finally, we prove the claim of the lemma.
	We have that $X_{i}^{b} = 0$ for all $b \in [1, i]$ since the buckets it refers to are frozen when fragment $F_{i}$ begins.
	Hence, it holds
		$\mathbb{E}[X_{i}] = \sum\limits_{b=i+1}^{m} \mathbb{E}[X_{i}^{b}]$,
	and hence for all $b \in [i+1, m]$ we have
		$\mathbb{E}[X_{i}^{b}] = \frac{\mathbb{E}[X_{i}]}{m-i}$.
\end{proof}

Let $x \in \mathcal{U}$ be an item that was fetched during phase $P$.
We say that this fetch is \emph{clean} if $x$ was not in the cache at the beginning of $P$, otherwise it is \emph{stale}.
Let $C$ be the number of clean fetches during~$P$ and $S$ be the number of stale fetches during $P$.
Note that the total cost during phase $P$ is then $C+S$.
Let $i \in [0, m-1]$, $C_{i}$ and $S_{i}$ are respectively the number of clean and stale fetches during fragment $F_{i}$.
Let $b \in [1, m]$, $S_{i}^{b}$ is the number of stale fetches during fragment $i$ in the $b$-th bucket to freeze --- we say that a stale fetch is \emph{in} a given bucket if the fetched item was previously evicted because that bucket was randomly chosen to perform the eviction.
Finally, $S^{b} = \sum_{i=0}^{m-1} S_{i}^{b}$ is the total number of stale fetches in the $b$-th bucket to freeze throughout the phase $P$.

\begin{lemma}
	\label{lem:ub_stale_fetches_in_one_bucket}
	Let $b \in [1, m]$, it holds that
        \[
        \mathbb{E}[S^{b}] \le \sum\limits_{i=0}^{b-1} \frac{\mathbb{E}[C_{i} + S_{i}]}{m-i}
        \]
\end{lemma}

\begin{proof}
	For a stale fetch to occur, the fetched item in question must have been evicted previously in the phase.
	Hence, for a given bucket and across the whole phase, the number of stale fetches is no greater than the number of evictions; more formally, $\forall b \in [1, m]$ it holds
	\begin{align*}
		S^{b} \le \sum\limits_{i=0}^{m-1} X_{i}^{b}
	\end{align*}
	Past fragment $F_{b-1}$, bucket $b$ becomes frozen and will experience no more evictions for the rest of the phase; it therefore holds that $\forall i \in [b, m-1], X_{i}^{b} = 0$.
	Using this remark and Lemma~\ref{lem:basic_upper_bound_without_bypassing}, we get
	\begin{align}
		\label{subclaim:1}
		\mathbb{E}[S^{b}] \le \sum\limits_{i=0}^{b-1} \mathbb{E}[X_{i}^{b}] = \sum\limits_{i=0}^{b-1} \frac{\mathbb{E}[X_{i}]}{m-i}
	\end{align}

	\pagebreak

	Moreover, after any eviction directly occurs a fetch: during any period of time, the number of evictions therefore is no larger than the number of fetches.
	Taking fragments as periods of time and noticing that any fetch is either clean or stale, it therefore holds that $\forall i \in [0, m-1],
		X_{i} \le C_{i} + S_{i}$.
	From this inequality we directly derive
	\begin{align}
		\label{subclaim:2}
		\sum\limits_{i=0}^{b-1} \frac{\mathbb{E}[X_{i}]}{m-i} \le \sum\limits_{i=0}^{b-1} \frac{\mathbb{E}[C_{i} + S_{i}]}{m-i},
	\end{align}
	which proves the claim combined with (\ref{subclaim:1}).
\end{proof}

\begin{lemma}
	\label{lem:main_ub}
	For any $b \in [1, m]$ it holds that
		$$\mathbb{E}[S^{b}] \leq \frac{\mathbb{E}[C]}{m-b+1}$$
\end{lemma}

\begin{proof}
 	By induction on $b$, we prove the following inequality:
	\begin{align}
		\label{ineq:induction_claim}
		\sum\limits_{i=0}^{b-1} \frac{\mathbb{E}[C_{i} + S_{i}] }{m-i} \leq \frac{1}{m-b+1} \left( \sum\limits_{i=0}^{b-1}  \mathbb{E}[C_{i}] - \sum\limits_{j=1}^{b-1} \sum\limits_{i=b}^{m-1} \mathbb{E}[S_{i}^{j}]
 \right)
	\end{align}
	which directly proves the claim along with using Lemma~\ref{lem:ub_stale_fetches_in_one_bucket}.
	Inequality (\ref{ineq:induction_claim}) clearly holds for $b=1$ since $S_{0} = 0$.
	Let $b \in [1, m]$, we suppose that the claim holds for $b$.
        Then,

	\begin{align*}
          \sum\limits_{i=0}^{b} \frac{\mathbb{E}[C_{i} + S_{i}]}{m-i} &= \sum\limits_{i=0}^{b-1} \frac{\mathbb{E}[C_{i} + S_{i}]}{m-i} + \frac{\mathbb{E}[C_{b}]+\mathbb{E}[S_{b}]}{m-b}\\
                &= \sum\limits_{i=0}^{b-1} \frac{\mathbb{E}[C_{i} + S_{i}]}{m-i} + \frac{\mathbb{E}[C_{b}] + \mathbb{E}[S^{b}] + \sum\limits_{j=1}^{b-1} \mathbb{E}[S_{b}^{j}] - \sum\limits_{i=b+1}^{m-1} \mathbb{E}[S_{i}^{b}]}{m-b}\\
          &\leq \sum\limits_{i=0}^{b-1} \frac{\mathbb{E}[C_{i} + S_{i}]}{m-i} + \frac{\mathbb{E}[C_{b}] + \sum\limits_{i=0}^{b-1} \frac{\mathbb{E}[C_{i} + S_{i}]}{m-i} + \sum\limits_{j=1}^{b-1} \mathbb{E}[S_{b}^{j}] - \sum\limits_{i=b+1}^{m-1} \mathbb{E}[S_{i}^{b}]}{m-b}\\
          &= \frac{m-b+1}{m-b} \cdot \left(\sum\limits_{i=0}^{b-1} \frac{\mathbb{E}[C_{i} + S_{i}]}{m-i}\right) + \frac{\mathbb{E}[C_{b}] + \sum\limits_{j=1}^{b-1} \mathbb{E}[S_{b}^{j}] - \sum\limits_{i=b+1}^{m-1} \mathbb{E}[S_{i}^{b}]}{m-b}\\
          &\leq \frac{\sum\limits_{i=0}^{b-1}  \mathbb{E}[C_{i}] - \sum\limits_{j=1}^{b-1} \sum\limits_{i=b}^{m-1} \mathbb{E}[S_{i}^{j}]}{m-b} + \frac{\mathbb{E}[C_{b}] + \sum\limits_{j=1}^{b-1} \mathbb{E}[S_{b}^{j}] - \sum\limits_{i=b+1}^{m-1} \mathbb{E}[S_{i}^{b}]}{m-b}\\
		&= \frac{1}{m-b} \cdot \left( \sum\limits_{i=0}^{b}  \mathbb{E}[C_{i}] - \sum\limits_{j=1}^{b} \sum\limits_{i=b+1}^{m-1} \mathbb{E}[S_{i}^{j}] \right)
	\end{align*}
        We used the identity $S_{b} = S^{b} + \sum_{j=1}^{b-1} S_{b}^{j} - \sum_{i=b+1}^{m-1} S_{i}^{b}$ at the line 2, Lemma~\ref{lem:basic_upper_bound_without_bypassing} at the line 3 and the induction hypothesis at the line 5.
	The induction holds hence the claim is true for all $b \in [1, m]$.
\end{proof}

\begin{theorem}
	\label{thm:main_up_without_bypassing}
	In each phase, \BucketALG fetches in expectation no more than $H_{\min \{k, \ell\}} \cdot C$ items, where $C$ is the number of clean fetches in the phase, $\ell$ is the size of the maximum independent set in the transitive closure of the dependency graph and $k$ is the size of the cache.
\end{theorem}

\begin{proof}
  We first examine the number of stale fetches in the last bucket to freeze.
  Until it freezes, they are no fetches in it (otherwise it would have frozen earlier); but when it finally freezes, the phase is over and no more fetches will occur.
	Hence, there are no stale fetches in the last bucket to freeze i.e., $S^{m} = 0$.
	Summing the inequalities from Lemma \ref{lem:main_ub} for all buckets $b \in [1,m-1]$ and taking into account the last bucket's specificity gives us the following
	\begin{align*}
		\mathbb{E}[S] \le \left( \sum\limits_{b=1}^{m-1} \frac{1}{m-b+1} \right) \cdot \mathbb{E}[C].
	\end{align*}
	Adding the expected total number of clean fetches on both members, noticing that $C$ does not depend on the random choices of \BucketALG (i.e., that $\mathbb{E}[C] = C$) and that the initial number of buckets $m$ is smaller than both $k$ and $\ell$ ends the proof.
\end{proof}

\subsubsection{A Lower Bound for any Algorithm in a Phase}
\label{sec:lowerbound_in_a_phase}

Let \OPT be an optimal algorithm for the input sequence $\sigma$.
Similarly to Fiat et al.~\cite{FiatKLMSY91}, we can prove that $\OPT$ has an (amortized) cost at least $C$ in every phase.
Only one technicality prevents us from directly applying their proof; it arises when the phase ends while \BucketALG have not yet fetched all the dependencies of the next request.
Then, there is no guarantee that \OPT must have paid to fetch those items during the current phase.
Intuitively, the lower-bound however still applies since those items will be paid by \OPT during the next phase anyway.

One method to present a rigorous proof is to compare with a specific optimal offline algorithm \OPTtau that performs the same evictions and fetches as \OPT for each request, but in a specific order: for each request, all the evictions happen first and then the fetches happen following the $\tau$ order.
\OPTtau pays the same cost as $\OPT$ and it is feasible.
More importantly, \OPTtau is synchronized with \BucketALG which allows to directly apply the proof from Fiat et al.

\subsubsection{Bounding the Competitive Ratio}
The following theorem ends the analysis of our randomized algorithm \BucketALG.

\begin{theorem}
  \label{thm:final_cr_no_bypassing}
  The algorithm \BucketALG is $2 H_{\min \{k, \ell\}}$-competitive against the oblivious adversary, where $\ell$ is the size of the maximum independent set in the transitive closure of the dependency graph and $k$ is the size of the cache.
\end{theorem}

\begin{proof}
  Directly from the above upper and lower-bounds (in Sections \ref{sec:lowerbound_in_a_phase} and 
\ref{sssec:ub_bucketalg}, respectively).

\end{proof}

This result is tight with the lower bound that we present in the next subsection.

\pagebreak
\subsection{A Lower Bound for Dependency-Aware Caching without Bypassing}
\label{apx:lb}

In this section we give a lower bound for the competitive ratio of any randomized online algorithm for caching with dependencies.
The competitive ratio depends on the topology of the graph, namely the maximum independent set (for the definition of the maximum independent set in directed graphs we refer to Section~\ref{ssec:notation}).
This lower bound asymptotically matches the upper bound given in Section~\ref{alg:random_mark_no_bypassing}.

\begin{theorem}
  No randomized online algorithm can achieve a competitive ratio better than $H_{\min\{k, \ell\}}$ against the oblivious adversary, where $k$ is the size of the cache and $\ell$ is the size of the maximum independent set of the dependency graph.
\end{theorem}

\begin{proof}
  
  Let $k$ and $\ell$ be any two integers.
  If $\ell \ge k$ then our claim can be directly derived by taking a~dependency graph without edges: in that case, it is well-known~\cite{FiatKLMSY91} that no randomized algorithm can perform better than $H_k$.
  We thus assume that $\ell < k$ in the following.
	
  We consider the following graph $G = (\mathcal{U}, E)$: $G$ has $n = k + 1$ nodes and consists of $\ell - 1$ isolated nodes and a single chain of length $k - \ell + 2$.
  Let $L$ be a maximum independent of $G$, $L$ has size $\ell$ and consists of the isolated nodes plus one node of the chain.
  
  We will prove that no randomized online algorithm has a better competitive ratio than $H_{\ell}$ for the dependency graph $G$ and cache size $k$.
  To this end, we use Yao's principle~\cite{Yao77}.
  We construct a probability distribution over input sequences as follows.
  First, issue a request to an item $r(0)$ from~$L$ uniformly at random, and follow with $k-\ell+1$ requests to all items from $\mathcal{U} \setminus L$.
  We continue by issuing requests to $r(i)$ chosen uniformly at random from $L \setminus \{ r(i-1) \}$, followed by $k-\ell+1$ requests to all items from $\mathcal{U}\setminus L$.
  
  Following Yao's principle, we consider any deterministic online algorithm on this instance.
  We lower bound its cost on a subsequence of requests $r(i)$ followed by $k-\ell+1$ requests to all items from $\mathcal{U}\setminus L$.
  Consider the algorithm's cache at the beginning of the subsequence: if the algorithm then has all the items of $L$ in its cache, it pays at least $1$ for some request from $\mathcal{U}\setminus L$.
  Otherwise, in a~case that the algorithm does not have all items of $L$ in its cache, then such a missing item will be requested with probability at least $1 / \ell$; the algorithm hence pays a cost of at least $1 / \ell$ in expectation for each subsequence.
  
  To bound the competitive ratio, we partition the request sequence into phases such that each phase contains $k+1$ requests to distinct pages.
  Any optimal offline algorithm incurs at most one cache miss per phase.
  To lower bound the cost of the algorithm, it remains to determine the expected length of the phase.
  The phase ends when all items from $L$ are requested in the phase.
  
  Consider a complete graph $K_\ell$ with $\ell$ nodes.
  The sequence of requests $r(i)$ in the phase corresponds to a random walk in $K_\ell$.
  The phase ends when the random walk visits all nodes of $K_\ell$.
  The expected number of steps to visit all nodes is at least $\ell \cdot H_\ell$.
  For each subsequence consisting of~$r(i)$ followed by requests to $\mathcal{U}\setminus L$, the algorithm pays at least $1/\ell$ and its expected cost for the phase thus is at least $H_\ell$.
  The claim holds since an optimal offline algorithm pays one per phase.
\end{proof}

\subsection{Deterministic Algorithm Recursive~LRU}
\label{sec:det}

To finish the section on the model variant without bypassing, we present a simple optimal deterministic algorithm.
This completes the picture for this model: both our randomized and deterministic results are tight.

We define a natural deterministic algorithm for online dependency-aware caching, recursive LRU, which invokes the classic LRU algorithm for all dependencies of the requested item.
Precisely, we split the algorithm's logic into two loops.
The first loop updates the timestamps top-to-bottom to assure that the evictions will take place in the correct order.
The second loop fetches dependencies bottom-to-top, to assure that at each intermediate cache state, all items have their dependencies in the cache, consistently with the model definition.

We claim that this algorithm is $k$-competitive.
The lower bounds for the classic online caching imply lower bounds for the more general dependency-aware caching problem, hence the lower bound of $k$ given by Sleator and Tarjan~\cite{Sleator1985} implies that \DET is optimal.

\begin{theorem}
	\DET is $k$-competitive.
	\label{thm:recursivelru}
\end{theorem}

\begin{proof}
	Let $\sigma$ be the input request sequence.
	We are going to compare the costs (i.e., number of fetches) of \DET and \OPTtau (defined in section \ref{sec:lowerbound_in_a_phase}) on $\sigma$ to prove the claim.

        For synchronization purposes between \DET and \OPTtau, we will pretend that the input sequence is not $\sigma$ but rather a more fine-grained sequence.
We call that fine-grained sequence $\sigma_{\tau}$, and construct it as follows.
Take the sequence $\sigma$, replace each item $v \in \sigma$ by its ancestors set $T(v)$ sorted by increasing $\tau$.
We refer to the items of $\sigma_{\tau}$ as \emph{pseudo-requests}, notice that both \DET and \OPTtau will correctly deal with $\sigma_{\tau}$ while performing the same actions as with $\sigma$ --- that is, performing the same fetches and evictions in the same order.

	We partition the sequence \sigtau into phases $P_{1}, P_{2}, \dots$ such that \DET fetches at most $k$ items while it serves the pseudo-requests of $P_{1}$ and exactly $k$ items while it serves the items of the other phases.
	Such a partitioning can be obtained easily.
	We start at the end of \sigtau and scan \sigtau backwards. 
	Whenever we have seen $k$ fetches made by \DET, we cut off a new phase.
	In the remainder of this proof, 
	we show that \OPTtau fetches at least one item in each phase.

	For phase $P_{1}$ there is nothing to show.
	Since \DET and \OPTtau start with the same cache, the first item that \DET fetches, the \OPTtau fetches as well.

	Consider an arbitrary phase $P_{i}$, $i \ge 2$.
	Let $x$ be the last item of the previous phase $P_{i-1}$.
	If \DET fetches $k$ distinct items that are different from $x$ while considering the items of $P_{i}$ then \OPTtau must fetch at least one item while it serves the pseudo-requests of $P_i$; this holds since \OPTtau has $x$ in its cache after it served the last pseudo-request of $P_{i-1}$ and thus cannot have all the other $k$ requested items in $P_{i}$ in its cache.

	Now suppose that $P_{i}$ does not contain $k$ distinct items that are different from $x$.
	We distinguish between two cases: either there exists an item $y$ that appears twice in $P_{i}$, or \DET fetches $x$ during~$P_{i}$.

	Let us first assume that \DET fetches an item $y$ twice in $P_{i}$, we note $t_{\textrm{fetch1}}$, $t_{\textrm{evict}}$ and $t_{\textrm{fetch2}}$ the indices of items in \sigtau that, when considered by \DET, respectively lead to the first fetch of $y$, its eviction and its second fetch.
	Let $v \in \sigma$ be the original request that lead to the first fetch of $y$ by \DET, it holds $y \in T(v)$.
	Let $S = \left\{ w \in T(v) : w <_\tau y\right\}$ be the set of items that are considered before $y$ in response to request $v$.
	At the moment $y$ is evicted, any other cached item $z$ has a more recent timestamp by definition of \DET.
	This could hold either because $z$ is in $S$ and therefore obtained a more recent timestamp than $y$ in the for loop of the line $2$, or because $z$ appears in \sigtau after the items of~$S$.
	Hence, for $y$ to be evicted, there were at least $k - |S| + 1$ different items requested between~$t_{\textrm{fetch1}}$ and $t_{\textrm{evict}}$.
	Since \OPTtau also has the items of $S$ in its cache at time $t_{\textrm{fetch1}}$, it must fetch an item.

	Finally, suppose that within $P_{i}$, \DET does not fetch twice the same item but fetches once the item $x$.
	In that case, we can apply the same reasoning as before between the last item of $P_{i-1}$ and the first occurrence of $x$ in $P_{i}$.

	For each phase, \DET pays at most $k$ while \OPTtau pays at least $1$, hence \DET is $k$-competitive.
\end{proof}

Given that the lower bound of $k$ for competitiveness of any deterministic algorithm for the classic caching problem~\cite{Sleator1985}, we conclude that \DET is an optimal deterministic algorithm for dependency-aware caching.

We note this simple design of recursively applying a known algorithm for the classic variant works for the deterministic case only, but does not lead to optimal competitiveness for the randomized algorithm.

\section{Dependency-Aware Caching with Bypassing}
\label{sec:bypassing}

In this section, we consider a model of dependency-aware caching with \emph{bypassing}, a similar setting to the classic caching with bypassing~\cite{EpsteinILN15,Irani02a}.
When a request arrives, an algorithm may choose to either (1) serve it from the cache by fetching the requested item and its missing dependencies, or (2) to bypass it for a cost of $1$.
The first option, serving a request from the cache, may be more costly if the number of the dependencies are large, however the subsequent requests to the same item are free (until the cache changes). 
The second option, bypassing a request, always cost $1$ per bypassed request, but does not require to fetch the dependencies of the requested item.

Introducing bypassing unravels new challenges to design a competitive online algorithm.
Whether it is more beneficial to fetch or to bypass a request depends on the future requests.
Hence, bypassing brings uncertainty to the online algorithm while it clearly benefits an offline algorithm.

In this section, we show how to use the randomized eviction procedure of the last section to take advantage of bypassing capabilities while providing strong worst case guaranties.
We present the algorithm \BucketALGBypass and prove it attains a~competitive ratio of $6 \sqrt{k \cdot H_{\min \{k, \ell\}}}$ against the oblivious adversary, where $k$ is the size of the cache, $\ell$ is the size of the maximum independent set in the transitive closure of the dependency graph and $H_k$ is the $k$-th Harmonic number.

\subsection{The Randomized Algorithm Bucketing with Bypassing}

\subsubsection{Presentation of the algorithm \BucketALGBypass}

We present our randomized algorithm \BucketALGBypass that handles the case when bypassing is allowed.
Just like \BucketALG, \BucketALGBypass operates in phases.
The first phase starts with the first request $\sigma_{1}$, and the phase ends when the algorithm resets its buckets.
At the start of a phase, \BucketALGBypass initializes a pool of buckets exactly like \BucketALG.

Whenever a request arrives, say to an item $v$, we distinguish between multiple cases.
If $v$ is already in the cache, \BucketALGBypass does the same as the algorithm \BucketALG without bypassing: it removes $T(v)$ from all buckets and serves the request.
Otherwise, in case $v$ is not in the cache, \BucketALGBypass selects a item $w$ successor of $v$ positioned just above the cache frontier; it chooses the one item with minimum $\tau$ to avoid ambiguity.
\BucketALGBypass then lists the items that are both successor of $w$ and in a bucket --- let $S$ be the set of those items.
If $S$ contains strictly more than $\sqrt{k / H_{\min \{k, \ell\}}}$ items then \BucketALGBypass removes the $\sqrt{k / H_{\min \{k, \ell\}}}$ ones with minimum $\tau$ from the buckets.
It bypasses $v$ and waits for the next request.
Otherwise, if $S$ contains fewer than $\sqrt{k / H_{\min \{k, \ell\}}}$ items then \BucketALGBypass removes those items from the buckets, applies the eviction procedure of \BucketALG if the cache is full and finally fetches $w$ into the cache.
If $w = v$, then \BucketALGBypass serves the request to $v$, else it bypasses it.
The algorithm uses the procedures $\text{ResetBuckets}$ and $\text{EvictAndFetch}$ defined in Algorithm \ref{alg:random_mark_no_bypassing}.
The pseudocode of the procedure \BucketALGBypass can be found in Algorithm \ref{alg:random_mark_with_bypassing}.

\begin{algorithm}
 \caption{\BucketALGBypass algorithm with bypassing}
 \label{alg:random_mark_with_bypassing}

\Alg{{\BucketALGBypass}}{
$\Gamma \gets \text{ResetBuckets}()$\\

\When{a new request $v$ arrives}
{
  \uIf{$v$ is in the cache}{
    Remove the items of $T(v)$ from all buckets\\
    remove all frozen buckets from $\Gamma$\\
    serve $v$\\
  }
  \uElse{
    Let $w$ be the uncached item of $T(v)$ with minimum $\tau$\\
    Let $S$ be the set of items that are simultaneously in $T(w)$ and in some bucket\\
    \uIf{$|S| > \sqrt{k / H_{\min \{k, \ell\}}}$}{
      Remove the $\sqrt{k / H_{\min \{k, \ell\}}}$ items of $S$ with minimum $\tau$ from all buckets in $\Gamma$\\
      Remove all frozen buckets from $\Gamma$\\
      bypass $v$\\
    }
    \uElse{
      Remove the items of $T(w)$ from all buckets\\
      Remove all frozen buckets from $\Gamma$\\
      
      \uIf{$\Gamma = \emptyset$}
      {
        $\Gamma \gets \text{ResetBuckets}()$\\
        Remove the items of $T(w)$ from all buckets\\
        Remove all frozen buckets from $\Gamma$\\
      }
      $\Gamma \gets \text{EvictAndFetch}(w, \Gamma)$\\
      If $w=v$ then serve request $v$ else bypass it\\
    }
  }
}
}
\end{algorithm}

\subsubsection{Upper-bound on the cost of \BucketALGBypass}
In the rest of this section, we bound the competitive ratio of \BucketALGBypass.
Like in the previous analysis without bypassing allowed, we define \emph{phases} as sequences of consecutive requests between two calls to the procedure $\text{ResetBuckets}$.

\pagebreak

\begin{restatable}{lemma}{mainUBBypassing}
  \label{lem:main_ub_bypassing}
  For each phase, \BucketALGBypass pays at most $2 H_{\min \{k, \ell\}} \cdot C + \sqrt{k \cdot H_{\min \{k, \ell\}}}$ in expectation, where $C$ is the number of clean requests during the considered phase.
\end{restatable}

\begin{proof}
  Let $P$ be a phase.
  We partition the requests of $P$ into two subsequences $P_{f}$ and $P_{b}$ depending on how \BucketALGBypass handles them.
  A request belongs to $P_{b}$ if the algorithm handles it with the lines $11$ to $13$, it belongs to $P_{f}$ otherwise.
  We now prove the two following inequalities which directly imply the claim:
  \begin{enumerate}
	\item \BucketALGBypass pays at most $\sqrt{k \cdot H_{\min \{k, \ell\}}}$ for the requests in $P_{b}$.
	\item \BucketALGBypass pays at most $2 H_{\min \{k, \ell\}} \cdot \mathbb{E}[C]$ for the requests in $P_{f}$
  \end{enumerate}
  We first prove the subclaim $(1)$.
  At the start of the phase the buckets together contain $k$ items, this number decreases until it hits zero at the end of the phase.
  Each time a request of $P_{f}$ is issued, \BucketALGBypass removes $\sqrt{k / H_{\min \{k. l\}}}$ items from the buckets.
  Hence, $P_{b}$ contains at most $k / \sqrt{k / H_{\min \{k. l\}}} = \sqrt{k \cdot H_{\min \{k. l\}}}$ requests.
  Since \BucketALGBypass pays $1$ for each request in $P_{b}$, it pays at most $\sqrt{k \cdot H_{\min \{k. l\}}}$.

  In order to prove the subclaim (2), we construct an alternative request sequence to $P_{f}$ and run \BucketALG on it.
  We show that \BucketALG does the same changes of internal variables and pays at least half the cost of \BucketALGBypass on $P_{f}$ dealing with that alternative request sequence.\\
  Let $v$ be a request of $P_{f}$.
  If $v$ is in the cache when requested, then \BucketALGBypass behaves like \BucketALG and pays $0$.
  Otherwise, if $v$ is not in the cache then everything behaves as if $w$ was the requested item in a framework without the possibility to bypass and fetching costs where doubled.
  Finally, the interferences of the requests in $P_{b}$ (that only shrink buckets) can be modeled by inserting in $P_{f}$ fake requests of cached items.
  We then apply Theorem \ref{thm:main_up_without_bypassing} and the subclaim follows.
\end{proof}

\subsection{A Lower Bound for Any Offline Algorithm and Bounding the Ratio}

Let \OPT be an optimal offline algorithm for the request sequence $\sigma$.
Based on \OPT, we define another offline algorithm \OPTtau as follows.
Let $\sigma_{i}$ be a request.
If \OPT bypasses $\sigma_{i}$ then \OPTtau bypasses it.
Otherwise, if \OPT does not bypass $\sigma_{i}$, \OPTtau fetches and evicts the same items as \OPT but in a specific order.
First, \OPTtau performs all the evictions that \OPT does.
Then, \OPTtau performs all the fetches that \OPT does sorted in increasing $\tau$.
\OPTtau maintains a feasible cache and pays the same cost as \OPT, \OPTtau is an optimal offline algorithm.\\

Similarly to the previous section, we compare the cost of \BucketALGBypass to the cost of \OPTtau.
The following claim holds.

\begin{restatable}{lemma}{optBypassLB}
	\label{lem:OPT_Bypass_lb}
	Let $P$ be a phase, $C$ is the number of clean fetches made by \BucketALGBypass during $P$. 
	Let $d_{I}$ and $d_{F}$ be the number of items in \OPTtau's cache not in \BucketALGBypass's cache at the beginning and at the end of the phase.
	It holds:
	\begin{align*}
		\text{\OPTtau}(P) \geq \frac{1}{2} \cdot \sqrt{\frac{H_{\min \{k, \ell\}}}{k}} \cdot (C - d_{I} + d_{F}).
	\end{align*}
\end{restatable}

The derived amortized lower bound on \OPTtau($P$) is proportional to the number of clean requests~$C$.
However, the upper bound on \BucketALGBypass($P$) in Theorem \ref{lem:main_ub_bypassing} features an additive term which does not depend on $C$.
As a last step before the final result, we give a second lower bound on \OPTtau($P$) in case $C$ is zero.

\begin{restatable}{lemma}{optWeakLB}
	\label{lem:OPT_Weak_Lower_Bound}
        Let $P$ be a phase.
	Let $\Phi_{I}$ and $\Phi_{F}$ be $1$ if \BucketALGBypass's and \OPTtau's caches are not the same, respectively at the beginning and at the end of the phase.
	It holds:
	\begin{align*}
		\text{\OPTtau}(P) \geq \frac{1}{2} \cdot (\mathbbm{1}_{\text{\BucketALGBypass}(P) > 0} - \Phi_{I} + \Phi_{F}).
	\end{align*}
\end{restatable}

\begin{proof}
  Without loss of generality, we assume that \OPTtau always has exactly $k$ items in its cache.
  Clearly, if both \BucketALGBypass and \OPTtau start the phase with the same cache then \OPTtau pays a~non-zero cost if \BucketALGBypass does, it hence holds $\text{\OPTtau}(P) \geq \mathbbm{1}_{\text{\BucketALGBypass}(P)>0} - \Phi_{I}$.
  Then, if \BucketALGBypass and \OPTtau end the phase with a different cache, this means that the requests that entailed the isolated cached items of \BucketALGBypass are not in \OPTtau's cache, forcing \OPTtau to pay at least $1$ to either bypass those requests or to cache the requested items and eventually evict them (which means \OPTtau fetches at least one item since it always has exactly $k$ items in its cache).
  It therefore holds that $\text{\OPTtau}(P) \geq \Phi_{F}$.
  Hence, it holds that 
  \begin{align*}
    \text{\OPTtau}(P) \geq \max \{\mathbbm{1}_{\BucketALGBypass>0} - \Phi_{I}, \Phi_{F}\} \geq \frac{1}{2} \cdot (\mathbbm{1}_{\BucketALGBypass>0} - \Phi_{I} + \Phi_{F}).
  \end{align*}
\end{proof}
\begin{theorem}

  The algorithm \BucketALGBypass is $6 \sqrt{k \cdot H_{\min \{k, \ell\}}}$-competitive against the oblivious adversary, where $\ell$ is the size of the maximum independent set in the transitive closure of the dependency graph and $k$ is the size of the cache.
\end{theorem}

\begin{proof}
	If \BucketALGBypass pays $0$ in a phase then the claimed competitive ratio holds.
	Otherwise, we derive the claimed competitive ratio using the two amortized costs of \OPTtau found in Lemmas~\ref{lem:OPT_Bypass_lb} and \ref{lem:OPT_Weak_Lower_Bound}.
	It finally holds that
	\begin{align*}
          \frac{\text{\BucketALGBypass}(P)}{\text{\OPTtau}(P)} &\leq \frac{2 H_{m} \cdot C + \sqrt{k \cdot H_{\min \{k, \ell\}}}}{1/2 \cdot \max \{\sqrt{\frac{H_{\min \{k, \ell\}}}{k}} \cdot C, \mathbbm{1}_{\text{\BucketALGBypass}(P)>0}\}}\\
                                                                           &\leq 2 \cdot \left(2 \sqrt{k \cdot H_{\min \{k, \ell\}}} + \sqrt{k \cdot H_{\min \{k, \ell\}}} \right)\\
                                                                           &\leq 6 \sqrt{k \cdot H_{\min \{k, \ell\}}}.
	\end{align*}
\end{proof}

\section{Evaluations}

While the main objective of this paper is to develop a rigorously analyzed online algorithm with formal bounds on its performance, we additionally perform preliminary empirical evaluations of the performance of our \BucketALGBypass algorithm.
The code of the evaluations is publicly available~\cite{repo}.
\pagebreak

\begin{figure*}[h]
    \centering
\subfigure[{\normalsize Zipf distribution of requests among nodes}]{
	\includegraphics[width=0.8\linewidth]{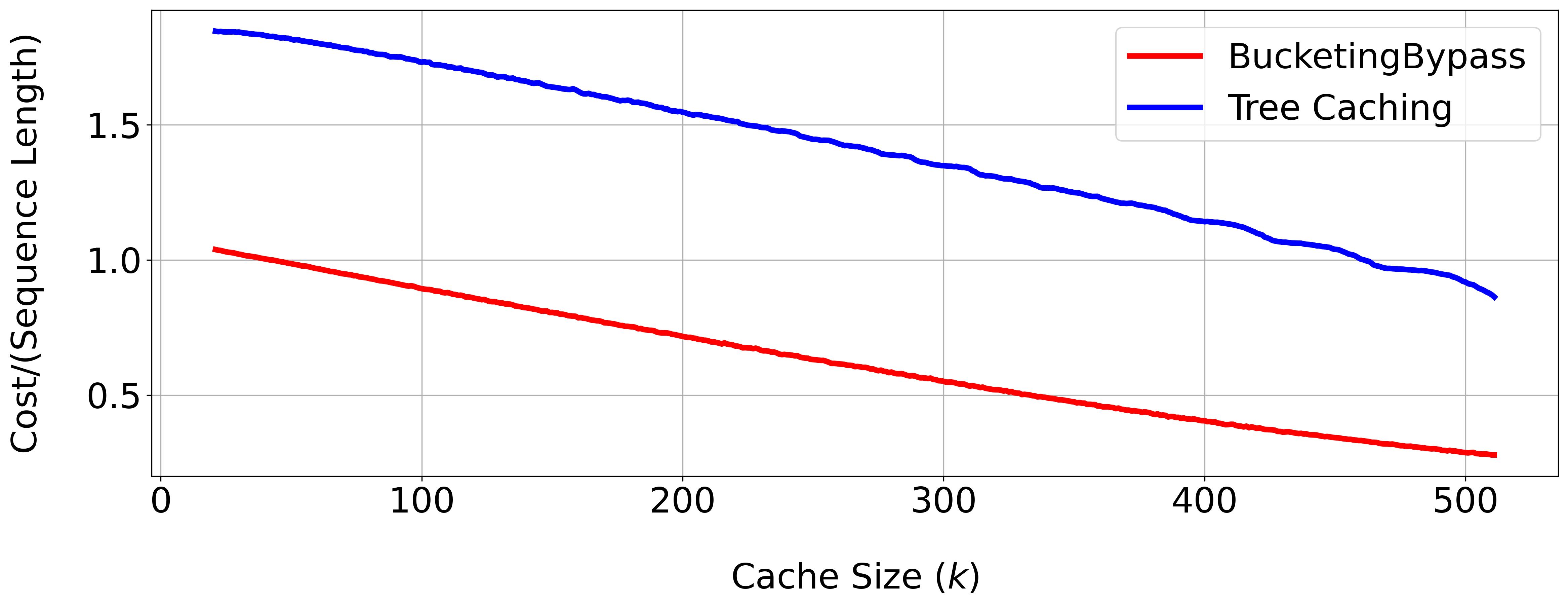}
}
    \subfigure[{\normalsize Exponential distribution of requests among nodes}]{
	\includegraphics[width=0.8\linewidth]{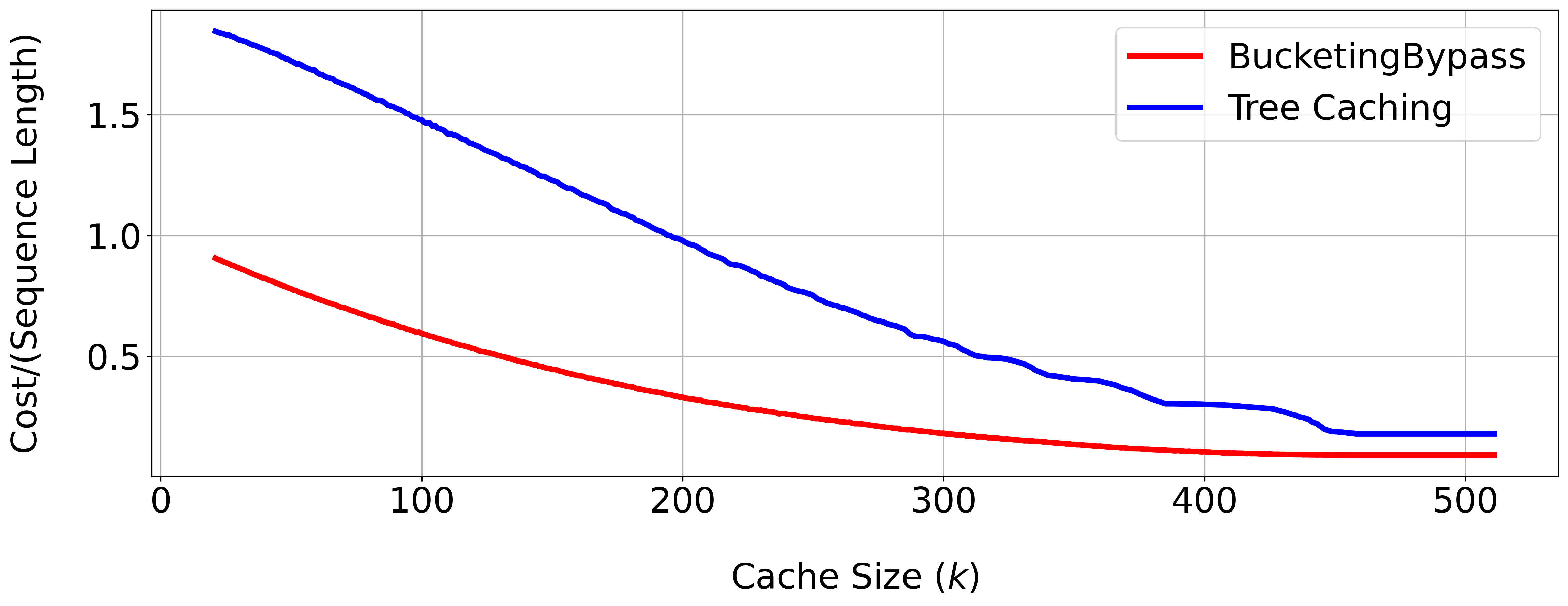}
}
    \caption{\normalsize Comparison of the \BucketALGBypass and Tree Caching algorithms in terms of cost per request for various cache sizes $k$ and a binary tree with height $h(T)=10$ and $1023$ nodes.
	The left subfigure plot uses Zipf distribution parameterized by $a=4$ and the right subfigure uses geometric distribution parameterized by $p = \frac{10}{2^{10}}$.
	}

 	\label{fig:costcachen10}
\end{figure*}

As a baseline, we compare with the algorithm TreeCaching of Bienkowski et al.~\cite{BienkowskiMPSS17}. 
Out of the two possible baselines, CacheFlow~\cite{KattaARW16} and TreeCaching~\cite{BienkowskiMPSS17}, the latter is closer in spirit to our solution, as both our system and TreeCaching optimize the same 
objective.
Notably, CacheFlow design does not account for the cost of changing the cache, but rather it may periodically exchange entire cache, whereas striking the perfect balance between the cost and the benefits of cache exchanges is a fundamental design principle of our algorithm.

Recall that our algorithm attains the competitive ratio of $6 \sqrt{k \cdot H_k}$, whereas TreeCaching attains the competitive ratio of $O(h(T)\cdot k)$, where $h(T)$ is the height of the tree. The analytical upper bounds favor our algorithm, with an improvement in terms of both parameters $h(T)$ and $k$, but the bounds concern the worst case, and in this section we compare these algorithms empirically.

To evaluate the algorithms, we use the following methodology. We construct a balanced binary tree with a height $h(T)$. To generate the requests, we conduct two experiments with different request distributions. The first experiment uses a Zipf distribution, a common distribution applied in this context~\cite{SarrarUFSH12}, and the second experiment uses an exponential distribution, modelling highly skewed request pattern.
We apply each distribution to the item heights: the height of each requested item is independently drawn from the distribution, resulting in higher probabilities for items at lower levels of the tree and lower probabilities for items at higher levels. Then, each for each level, an item from that level is chosen uniformly at random.
This method may generate requests to items that can never fit in the cache, as they have more descendants than~$k$. We prune these requests from our request sequence using rejection sampling.

For the first experiment, we use the Zipf distribution.
The probability of an item with depth $1\le i \le h(T)$ being requested, where $i=1$ corresponds to the top item of the tree, is given by
\[
	Pr(i) = \frac{1}{2^{i-1}}\frac{(h(T) - i + 1)^{-a}}{\zeta(a)},
\]
where $\zeta$ represents the Riemann Zeta function, and we use the parameter value $a = 4$.
For the second experiment, we use the geometric distribution.
Requesting an item with index $1 \le i \le 2^{h(T)} - 1$ has the following probability
\[
	Pr(i) = (1 - p)^{i-1} p,
\]
where we use $p = \frac{10}{2^{10}}$.

We generate a random sequence of requests with a fixed length $5000$ for both distributions. Then, we calculate the total cost of the \BucketALGBypass and TreeCaching algorithms on such sequence, for various cache sizes $k$. We repeat each sequence 10 times and take an average cost, and then we determine the average cost per request for both algorithms.

In Figure \ref{fig:costcachen10}, we present the results of our evaluations of the mean cost per request between the \BucketALGBypass and TreeCaching algorithms.  
The \BucketALGBypass algorithm consistently induces significantly lower costs compared to the TreeCaching algorithm across all inspected cache sizes $k$. Our algorithm incurs on average 2x lower cost per requests than TreeCaching.

\section{Conclusions and Future Directions}

This paper proposed several competitive caching algorithms that are aware of dependencies among items.
We leave interesting avenues open for future studies.
Most notably, it will be interesting to study dependency aspects in more general
variants of caching, such as weighted caching~\cite{CohenK99,Young94,BansalBN12} or file caching ~\cite{Young02,AdamaszekCER19}.

To support further research and ensure reproducibility, we made our code and evaluation artefacts available at \url{https://github.com/foo/dependency-aware-caching}.

\vspace{0.2cm}
\textbf{Acknowledgements.}
This work was supported by the Austrian Science Fund (FWF) project I~5025-N (DELTA).

\bibliographystyle{IEEEtran}
\balance
\bibliography{references} 
\balance

\end{document}